\documentclass[12pt,a4paper,reqno]{article}
\usepackage{graphics}
\usepackage{epsfig}

\usepackage{amsmath,amsthm,amssymb}
\newtheorem{theorem}{Theorem}[section]
\newtheorem{corollary}[theorem]{Corollary}
\newtheorem{lemma}[theorem]{Lemma}

\newtheorem{example}[theorem]{Example}
\theoremstyle{definition}

\theoremstyle{remark}

\numberwithin{equation}{section}

\begin{document}
\title{The Structure of $\mathbb{Z}_{2}[u]\mathbb{Z}_{2}[u,v]$-additive codes
\thanks{The first author would like to thank the Department of Science and Technology (DST), New Delhi, India for their financial support
in the form of INSPIRE Fellowship
(DST Award Letter No.IF130493/DST/INSPIRE Fellowship / 2013/(362) Dated: 26.07.2013) 
to carry out this work.}
}
\author{N. Annamalai\\
Research Scholar\\
Department of Mathematics\\
Bharathidasan University\\
Tiruchirappalli-620 024, Tamil Nadu, India\\
{Email: algebra.annamalai@gmail.com}
\bigskip\\
C.~Durairajan\\
Assistant Professor\\
Department of Mathematics\\ 
School of Mathematical Sciences\\
Bharathidasan University\\
Tiruchirappalli-620024, Tamil Nadu, India\\
{Email: cdurai66@rediffmail.com}
\hfill \\
\hfill \\
\hfill \\
\hfill \\
{\bf Proposed running head:} The Structure of $\mathbb{Z}_{2}[u]\mathbb{Z}_{2}[u,v]$-additive codes}
\date{}
\maketitle

\newpage

\vspace*{0.5cm}
\begin{abstract}In this paper,  we study  the algebraic structure of $\mathbb{Z}_{2}[u]\mathbb{Z}_{2}[u,v]$-additive codes which are $\mathbb{Z}_{2}[u,v]$-submodules
where $u^2=v^2=0$ and $uv=vu$. In particular, we determine
a Gray map from $\mathbb{Z}_{2}[u]\mathbb{Z}_{2}[u,v]$ to $\mathbb{Z}_{2}^{2\alpha+8\beta}$ and study generator and parity check matrices 
for these codes. Further we study the structure of $\mathbb{Z}_{2}[u]\mathbb{Z}_{2}[u,v]$-additive cyclic codes and constacyclic codes.
\end{abstract}
\vspace*{0.5cm}

\vspace{0.5cm}
\vspace{1.5cm}

\noindent
Corresponding author:\\ 
\\
 \hspace{1cm} 
Dr. C. Durairajan\\
\noindent
Assistant Professor\\
Department of Mathematics \\
Bharathidasan University\\
Tiruchirappalli-620024, Tamil Nadu, India\\
\noindent
E-mail: cdurai66@rediffmail.com
\newpage

\newpage

\section{Introduction}
Let $\mathbb{Z}_{2}$ be the ring of integers modulo $2.$ Let  $\mathbb{Z}_{2}^{n}$ denote the set of all binary vectors of lenght $n.$ Any non-empty
 subset of $\mathbb{Z}_{2}^{n}$ is a binary code of length $n$ and a subgroup of $\mathbb{Z}_{2}^{n}$ is called a binary 
 linear code. In this paper, we introduce a 
 subfamily of binary linear codes, called $\mathbb{Z}_{2}[u]\mathbb{Z}_{2}[u,v]$-linear codes.
 
 Additive codes were first defined by Delsarte in $1973$ in terms of association schemes \cite{Bro}, \cite{Hamm}. In a translation association scheme, an additive code 
 is generally defined as a subgroup of the underlying abelian group.

Codes over finite rings were studied in \cite{andr}, \cite{eug} and with the remarkable paper by Hammons $\it{et .al} \cite{Bro}$ have been studied intensively for the last four decades. A study on mixed
alphabet codes has been introduced and some bounds have been presented by Brouwer $\it{et. al}.\cite{Ayod}.$ 

Later, $\mathbb{Z}_{2}\mathbb{Z}_{4}$-additive codes were generalized to $\mathbb{Z}_{2}\mathbb{Z}_{2^s}$-additive codes by Aydogu and Siap
\cite{Ayo}. For $s\geq 2,$ these generalisations to $\mathbb{Z}_{2}\mathbb{Z}_{2^s}$-additive codes are interesting since they provided good binary codes via
Gray maps with rich algebraic structure.  Furthermore, the structure of 
$\mathbb{Z}_{p^{r}}\mathbb{Z}_{p^{s}}$-additive codes and their duals has been determined \cite{ayo}. In this correspondance, begin inspired by 
these additive codes, a generalization towards another direction that have a
good algebraic structure and provide good binary codes is presented.

The aim of this paper is the study of the algebraic structure of $\mathbb{Z}_{2}[u]\mathbb{Z}_{2}[u,v]$-additive codes and their dual codes. 
It is organized as follows. In Section II, we recall the  linear codes over the ring $\mathbb{Z}_{2}+u\mathbb{Z}_{2}$ and
the linear codes over the ring $\mathbb{Z}_{2}+u\mathbb{Z}_{2}+v \mathbb{Z}_{2}+uv\mathbb{Z}_{2}.$ In Section III, we study the concept of 
$\mathbb{Z}_{2}[u]\mathbb{Z}_{2}[u,v]$-additive codes and duality, we determine the generator matrices of the these codes. In section IV, we study
the structure of $\mathbb{Z}_{2}[u]\mathbb{Z}_{2}[u,v]$-additive cyclic codes. In section V, we study
the structure of $\mathbb{Z}_{2}[u]\mathbb{Z}_{2}[u,v]$-additive constacyclic codes.

\section{Linear Codes over Chain Rings}
Let $R$ be a ring, then $R^{n}$ is an $R$-module. A submodule of $R^{n}$ is called $R$-linear code. We denote the Hamming weight of a word
as $w_{H}$ and Lee weight as $w_{L}.$

 The set $\mathbb{Z}_{2}+u\mathbb{Z}_{2}=\{0,1,u,1+u\},$ where $u^{2}=0,$ is a commutative finite chain ring of four elements. 
 We denote this ring as $R_{1}.$

 In \cite{ash}, a $R_{1}$-linear code $C$ is permutation equivalent to a code with generator matrix
 
 \begin{equation*}
G=
 \begin{bmatrix}
   I_{l_{0}}&A&B_{1}+uB_{2}\\
   0&uI_{l_{1}}&uD
  \end{bmatrix},
 \end{equation*}
where $A,B_{1},B_{2}$ and $D$ are matrices over $\mathbb{Z}_{2}.$

Define
$w_{L}(a+ub)=w_{H}(b,a\oplus b)$ for all $a,b\in \mathbb{Z}_{2}$ where the Hamming weight of a word is the number of non-zero coordinates in it.

The definition of the weight immediately leads to a Gray map from $R_{1}$ to $\mathbb{Z}_{2}^2$ which can
naturally be extended to $R_{1}^n$:

$\chi_{L}:(R_{1}^n, \text{Lee weight}) \rightarrow (\mathbb{Z}_{2}^{2n}, \text{Hamming weight})$ is defined
as:
$$\chi_{L}(\bar{a}+u\bar{b})=(\bar{b},\bar{a}\oplus \bar{b}),$$ 
where $\bar{a},\bar{b}\in\mathbb{Z}_{2}^n$ and $\oplus$ denotes addition in $\mathbb{Z}_{2}.$

Let $R_{2}$ be the ring $\mathbb{Z}_{2}[u,v]:=\mathbb{Z}_{2}+u\mathbb{Z}_{2}+v\mathbb{Z}_{2}+uv\mathbb{Z}_{2}$ with $u^2=0,\,\, v^2=0$
 and $uv=vu$ where $\mathbb{Z}_{2}=\{0,1\}.$ Then 
 $R_{2}$ is a commutative finite chain ring of $16$ elements and 
  $a+ub+vc+uvd\in R_{2}$ is a unit iff $a\neq 0.$
 
The group of units of $R_{2}$ is given by
  $R_{2}^*=\{1,1+u,1+v,1+uv,1+u+v,1+u+v+uv,,1+u+uv,1+v+uv\}.$
It is also a local ring with the unique maximal ideal $uvR_{2}.$

 A linear code $C$ of lenght $n$ over the ring $R_{2}$ is an $R_{2}$-submodule of $R_{2}^n.$

 A non-zero $R_{2}$-linear code $C$  has a generator matrix which after a suitable permutation of the coordinates can be written in the form:
 \begin{equation}
  G=
   \begin{bmatrix}

I_{k_{0}}&A_{01}&A_{02}&A_{03}&A_{04}\\
 0&uI_{k_{1}}&uA_{12}&uA_{13}&uA_{14}\\
  0&0&vI_{k_{2}}&vA_{23}&vA_{24}\\
 0&0&0&uvI_{k_{3}}&uvA_{34}
\end{bmatrix}
  \end{equation}
where $A_{ij}$ are all matrices over $\mathbb{Z}_{2}.$

To define the Lee weights and Gray maps for codes over $R_{2}$. 
Define
$w_{L}(a+ub+vc+uvd)=w_{H}(d,a\oplus d,b\oplus d,a\oplus b\oplus d, c\oplus d,a\oplus c\oplus d,
b\oplus c\oplus d,a\oplus b\oplus c\oplus d)$ for all $a,b,c,d\in \mathbb{Z}_{2}.$
The definition of the weight immediately leads to a Gray map from $R_{2}$ to $\mathbb{Z}_{2}^8$ which can naturally be extended to $R^n$:

$\phi_{L}:(R_{2}^n, \text{Lee weight})
\rightarrow (\mathbb{Z}_{2}^{8n}, \text{Hamming weight})$ 

is defined as:
\begin{equation*}
 \begin{split}
  \phi_{L}(\bar{a}+u\bar{b}+v\bar{c}+uv\bar{d})&=(\bar{d},\bar{a}\oplus \bar{d},\bar{b}\oplus \bar{d},\bar{a}\oplus \bar{b}\oplus \bar{d},
\bar{c}\oplus \bar{d},\\
&\bar{a}\oplus \bar{c}\oplus \bar{d},\bar{b}\oplus \bar{c}\oplus \bar{d},\bar{a}\oplus \bar{b}\oplus \bar{c}\oplus \bar{d})
 \end{split}
\end{equation*}

where $\bar{a},\bar{b},\bar{c},\bar{d}\in\mathbb{Z}_{2}^n.$

\section{$\mathbb{Z}_{2}[u]\mathbb{Z}_{2}[u,v]$-additive codes}

We know that the ring $\mathbb{Z}_{2}[u]$ is a subring of the ring $R_{2}$. Being inspired by the structure of $\mathbb{Z}_{2}\mathbb{Z}_{2}[u]$-additive codes,
we define the following set:\\
\begin{center}
 $\mathbb{Z}_{2}[u]\mathbb{Z}_{2}[u,v]=\{(a,b)\,\,\, | \,\,\, a\in \mathbb{Z}_{2}[u] \,\, \text{and} \,\, b\in \mathbb{Z}_{2}[u,v]\}.$
\end{center}

The set $\mathbb{Z}_{2}[u]\mathbb{Z}_{2}[u,v]$ cannot be endowed with algebraic structure directly. It is not well defined with respect to the 
usual scalar multiplication by  $v\in R_{2}.$ Therefore, this set is not an $R_{2}$-module. To make it well defined and enrich with an algebraic
structure we introduce a new scalar multiplication as follows:
We define a mapping
\begin{center}
 $\eta: R_{2}\rightarrow \mathbb{Z}_{2}[u],$\\
 $\eta(a+ub+vc+uvd)=a+ub$\\
i.e., $\eta(x)=x \mod v.$
\end{center}
 It is easy to conclude that $\eta$ is a ring homomorphism. Using this map, we define a scalar multiplication. For
$v=(a_{0},a_{1},\cdots,a_{\alpha-1},b_{0},b_{1},\cdots,b_{\beta-1})\in \mathbb{Z}_{2}[u]^{\alpha} \times
R_{2}^\beta$
and  
$l\in R_{2},$ we have
\begin{equation}\label{a}
 lv=(\eta(l)a_{0},\eta(l)a_{1},\cdots,\eta(l)a_{\alpha-1},lb_{0},lb_{1},\cdots,lb_{\beta-1}).
\end{equation}

 A non-empty subset $C$ of  $R_{1}^{\alpha} \times R_{2}^\beta$ is called \textbf{$\mathbb{Z}_{2}[u]\mathbb{Z}_{2}[u,v]$-additive code} if it is a
 $R_{2}$-submodule of 
 $R_{1}^{\alpha} \times R_{2}^\beta$ with respect to the scalar 
 multiplication defined in Eq.\ref{a}.
 Then the binary image $\Phi(C)=D$ is called $\mathbb{Z}_{2}[u]\mathbb{Z}_{2}[u,v]$-linear code of length $n=2\alpha+8\beta$ 
 where $\Phi$ is a map
 from $R_{1}^{\alpha}\times R_{2}^\beta$ to $\mathbb{Z}_{2}^n$ defined as
 \begin{equation*}
 \begin{split}
  \Phi(x,y)&=(q_{0},q_{1},\cdots,q_{\alpha-1}, q_{0}+r_{0},\cdots,q_{\alpha-1}+r_{\alpha-1},\\
  &d_{0},d_{1},\cdots,d_{\beta-1},a_{0}+d_{0},a_{1}+d_{1},\cdots,a_{\beta-1}+d_{\beta-1},\\
  &b_{0}+d_{0},b_{1}+d_{1},\cdots,b_{\beta-1}+d_{\beta-1},\\
    &a_{0}+b_{0}+d_{0},a_{1}+b_{1}+d_{1},\cdots,a_{\beta-1}+b_{\beta-1}+d_{\beta-1},\\
  &c_{0}+d_{0},c_{1}+d_{1},\cdots,c_{\beta-1}+d_{\beta-1},\\
  &a_{0}+c_{0}+d_{0},a_{1}+c_{1}+d_{1},\cdots,a_{\beta-1}+c_{\beta-1}+d_{\beta-1},\\
  &b_{0}+c_{0}+d_{0},b_{1}+c_{1}+d_{1},\cdots,b_{\beta-1}+c_{\beta-1}+d_{\beta-1},\\ 
  &a_{0}+b_{0}+c_{0}+d_{0},\cdots,a_{\beta-1}+b_{\beta-1}+c_{\beta-1}+d_{\beta-1}),  
 \end{split}
 \end{equation*}
for all $x=(x_{0},x_{1},\cdots,x_{\alpha-1})\in R_{1}^\alpha$ where $x_{i}=r_{i}+uq_{i}$  for $0\leq i\leq \alpha-1$ and 
$y=(y_{0},y_{1},\cdots,y_{\beta-1})\in R_{2}^\beta$, where $y_{j}=a_{j}+ub_{j}+vc_{j}+uvd_{j}$ for $0\leq j\leq \beta-1.$

 If $C\subseteq R_{1}^{\alpha}\times R_{2}^{\beta}$ is a $\mathbb{Z}_{2}[u]\mathbb{Z}_2[u, v]$-additive code, group isomorphic to 
 $\mathbb{Z}_{2}^{2l_{0}}\times \mathbb{Z}_{2}^{l_{1}}\times \mathbb{Z}_2^{4k_{0}}\times \mathbb{Z}_2^{3k_{1}}\times 
 \mathbb{Z}_2^{2k_{2}}\times \mathbb{Z}_{2}^{k_{3}},$
 then $C$ is called a
 \textbf{$\mathbb{Z}_{2}[u]\mathbb{Z}_2[u, v]$-additive code of type $(\alpha, \beta;l_{0},l_{1}; k_{0},k_{1},k_{2},k_{3})$ }
 where $l_{0},l_{1},k_{0}, k_{1},k_{2},$ and $k_{3}$  are defined above.

\subsection{Generator matrices of $\mathbb{Z}_{2}[u]\mathbb{Z}_{2}[u,v]$-additive codes}

\begin{theorem}
 Let $C$ be a $\mathbb{Z}_{2}[u]\mathbb{Z}_2[u, v]$-additive code of type  $(\alpha, \beta;l_{0},l_{1}; k_{0},k_{1},k_{2},k_{3})$. 
 Then $C$ is permutation equivalent to a  $\mathbb{Z}_{2}[u]\mathbb{Z}_2[u, v]$-additive code with standard form matrix.
 
 \begin{equation}\label{b}
 G=
 \begin{bmatrix}
 B&T\\
 S&A
 \end{bmatrix}
  \end{equation}
  where
  \begin{equation*}
  \begin{split}
 B&=
 \begin{bmatrix}
I_{l_{0}}&B_{01}&B_{02}\\
0&uI_{l_{1}}&uB_{12}
 \end{bmatrix}\\  
 T&=
 \begin{bmatrix}
 0&0&0&vT_{01}&vT_{02}\\
 0&0&0&0&uvT_{12}
 \end{bmatrix}\\
 S&=
 \begin{bmatrix}
 0&S_{01}&S_{02}\\
 0&0&uS_{12}\\
 0&0&0\\
 0&0&0
 \end{bmatrix}\\
  A&=
   \begin{bmatrix}
I_{k_{0}}&A_{01}&A_{02}&A_{03}&A_{04}\\
 0&uI_{k_{1}}&uA_{12}&uA_{13}&uA_{14}\\
  0&0&vI_{k_{2}}&vA_{23}&vA_{24}\\
 0&0&0&uvI_{k_{3}}&uvA_{34}
\end{bmatrix}
\end{split}
 \end{equation*}.

  Here, $B_{ij}$ are matrices over $R_{1}$ and $T_{ij}$ are matrices over
  $R_{2}$ for $0\leq i<2$ and $0<j\leq 2$. Further, 
  for $0\leq d < 2,\,\, 0\leq t <4,\,\, 0<q\leq 4,$ $S_{dj}$ and $A_{tq}$ are matrices over
  $R_{2}$.
   Also, $I_{l_{w}}$ and $I_{k_{y}}$ are the identity matrices of size $l_{w}$ and $k_{y}$, where $0\leq w\leq 1$ and $0\leq y\leq 3.$ 
\end{theorem}
\begin{proof}
 First, by focusing on the first $\alpha$ coordinates which is a projection of the $\mathbb{Z}_{2}[u]\mathbb{Z}_{2}[u,v]$-additive code to a linear
 $\mathbb{Z}_{2}[u]$-code and applying the necessary row operations together with column operations(if neessary) we obtain the matrix $B$ as
 in the upper left corner of the matrixEq.\ref{b}. Next, by focusing on the last $\beta$ coordinates and protecting the form of the matrix $B,$ 
 one consider linear codes over the ring $\mathbb{Z}_{2}[u,v]$ for which the generator matrices are well known to be of the form $A$ as in Eq.\ref{b}.
  Finally, by protecting both the form of $G$ as in Eq.\ref{b} which is referred to as the standard form matrix.
 
\end{proof}

\begin{example}
 Let $C$ be a $\mathbb{Z}_{2}[u]\mathbb{Z}_{2}[u,v]$-additive code of type $(2,3;1,1;1,1,1,0)$ with standard
 generator matrix of $C$ is
 \begin{equation}\label{c}
  \begin{bmatrix}
   1&1&0&\vline&0&0&0&v&v\\
   0&u&u&\vline&0&0&0&0&uv\\
   \hline
   0&1&0&\vline&1&0&1&0&1\\
   0&0&u&\vline&0&u&u&u&u\\
   0&0&0&\vline&0&0&v&v&v\\
   0&0&0&\vline&0&0&0&0&uv
  \end{bmatrix}
 \end{equation}
 
  and $C$ has $|C|=2^{2}2^{1}2^{4}2^{3}2^{2}=2^{12}=4096$ codewords.
\end{example}

\subsection{Duality of $\mathbb{Z}_{2}[u]\mathbb{Z}_{2}[u,v]$-additive codes}
An inner product for two elements $x,y\in R_{1}^\alpha\times R_{2}^\beta$ is defined as\\
\begin{equation*}
 \langle x,y\rangle=uv\left(\sum\limits_{i=1}^{\alpha}x_{i}y_{i}\right)+\sum\limits_{j=\alpha+1}^{\alpha+\beta}x_{j}y_{j}\in 
 R_{2}.
\end{equation*}
Let $D$ be a $\mathbb{Z}_{2}[u]\mathbb{Z}_{2}[u,v]$- linear code of length $n$ of type $(\alpha,\beta;l_{0},l_{1};k_{0},k_{1},k_{2},k_{3})$. 
We denote by $C$ the corresponding additive code, i.e., $C=\Phi^{-1}(D)$. The additive dual code of $C,$ denoted by $C^{\perp},$ is defined in the 
standard way
\begin{center}
 $C^{\perp}=\{y\in R_{1}^\alpha\times R_{2}^\beta \,|\, \langle x,y\rangle=0 \,\,\text{for all}\,\, x\in C\}$
\end{center}
The corresponding binary code $\Phi(C^{\perp})$ is denoted by $C_{\perp}$ and is called the $\mathbb{Z}_{2}[u]\mathbb{Z}_{2}[u,v]$-
dual code of $C.$

The additive dual code $C^{\perp}$ also an additive code, that is a subgroup of $R_{1}^{\alpha}\times R_{2}^{\beta}.$ 
\begin{theorem}
 Let $C$ be a $\mathbb{Z}_{2}[u]\mathbb{Z}_{2}[u,v]$-linear code of type $(\alpha,\beta;l_{0},l_{1};k_{0},k_{1},k_{2},k_{3}).$ The 
 $\mathbb{Z}_{2}[u]\mathbb{Z}_{2}[u,v]$-dual code $C_{\perp}$ is then of type $(\alpha,\beta;l_{0}^{'},l_{1}^{'};k_{0}^{'},k_{1}^{'},k_{2}^{'},k_{3}^{'}),$
 where
\begin{equation*}
 \begin{split}
  l_{0}^{'}&=\alpha-l_{0}-l_{1}\\
  l_{1}^{'}&=l_{1}\\
  k_{0}^{'}&=\beta-k_{0}\\
  k_{1}^{'}&=\beta-k_{1}-k_{2}\\
  k_{2}^{'}&=k_{2}\\
  k_{3}^{'}&=\beta-k_{2}-k_{3}
 \end{split}
\end{equation*}
\end{theorem}

\subsection{Parity-check matrices of  $\mathbb{Z}_{2}[u]\mathbb{Z}_{2}[u,v]$-additive codes}
As  in the classical case, the generator matrix of the dual code is important. In the following theorem, the standard form of the generator
matrix of the dual code is presented.

\begin{theorem}
 Let $C$ be a $\mathbb{Z}_{2}[u]\mathbb{Z}_{2}[u,v]$-additive code of type $(\alpha,\beta;l_{0},l_{1};k_{0},k_{1},k_{2},k_{3})$ 
 with the standard form matrix
 defined in equation (\ref{b}). Then the generator matrix for the additive dual code $C^{\perp}$ is given by
 
\begin{equation}
H=
\begin{bmatrix}
 \bar{B}&U\\
 V&\bar{A}+E
 \end{bmatrix}
 \end{equation}
 
 where
 
 \begin{equation*}
 \begin{split}
\bar{B}&=
\begin{bmatrix}
 -B_{02}^{t}+B_{12}^{t}B_{01}^{t}&-B_{12}^{t}&I_{\alpha-l_{0}-l_{1}}\\
 -uB_{01}^{t}&uI_{l_{1}}&0
 \end{bmatrix}\\ 
V&=
\begin{bmatrix}
 -T_{02}^{t}+T_{12}^{t}B_{01}^{t}+A_{34}^{t}T_{01}^{t}& -T_{12}^{t}&0\\
 -uT_{01}^{t}&0&0\\
 0&0&0\\
 0&0&0
 \end{bmatrix}\\
U&=
\begin{bmatrix}
 -uS_{12}^{t}A_{01}^{t}+u(B_{12}^{t}S_{01}^{t}-S_{02}^{t})& -uS_{12}^{t}&0&0&0\\
 -uvS_{01}^{t}&0&0&0&0
 \end{bmatrix}\\
  E&=
  \begin{bmatrix}
   T_{12}^{t}S_{01}&0&0&0&0\\
   0&0&0&0&0\\
   0&0&0&0&0\\
   0&0&0&0&0
  \end{bmatrix}\\    
\bar{A}&=
\begin{bmatrix}
l&-um&-uA_{12}^{t}&-uA_{23}^{t}&uI_{\beta-k_{1}-k_{2}}&0\\
 -vp&-vA_{12}^{t}&vI_{k_{2}}&0&0\\
-uvA_{01}^{t}&uvI_{\beta-k_{2}-k_{3}}&0&0&0
 \end{bmatrix}
 \end{split}
 \end{equation*}
 
 where 
 
$ l=-A_{04}^{t}+A_{01}^{t}(A_{14}^{t}+A_{13}^{t}A_{34}^{t}+A_{24}^{t}A_{12}^{t}+A_{12}^{t}A_{23}^{t}A_{34}^{t})+A_{02}^{t}(A_{24}^{t}+A_{23}^{t}A_{34}^{t})+
 A_{03}^{t}A_{34}^{t}$\\

 $k=A_{14}^{t}+A_{13}^{t}A_{34}^{t}+A_{24}^{t}A_{12}^{t}+A_{12}^{t}A_{23}^{t}A_{34}^{t}$\\
 
 $n=A_{24}^{t}+A_{23}^{t}A_{34}^{t}$\\
 
 $m=A_{03}^{t}+A_{02}^{t}A_{23}^{t}+A_{01}^{t}A_{12}^{t}$\\
 
 $p=A_{02}^{t}+A_{01}^{t}A_{12}^{t}$

\end{theorem}
\begin{proof}
 It is easy to check that $HG^{t}=0.$ Besides, since we have $|C||C^{\perp}|=2^{2\alpha+8\beta}$ by checking the type of the matrix $H$ we
 conclude that the rows of $H$ not only are orthogonal to the rows of $G$ but also generate the all code $C^{\perp}.$ Hence, $H$ is the desired matrix. 
\end{proof}

\begin{example}
 Let $C$ be a $\mathbb{Z}_{2}[u]\mathbb{Z}_{2}[u,v]$- additive code of type $(2,3;1,1;1,1,1,0)$ with the generator matrix given in Equation (\ref{c}).
 Then, we can write the parity-check matrix of $C$  as follows:
 \begin{equation*}
H=
 \begin{bmatrix}
  1&1&1&\vline&u&u&0&0&0\\
  u&u&0&\vline&uv&0&0&0&0\\
  \hline
  1&1&0&\vline&0&0&0&1&1\\
  u&0&0&\vline&u&u&u&u&0\\
  0&0&0&\vline&v&v&v&0&0\\
  0&0&0&\vline&0&uv&0&0&0
 \end{bmatrix}.
\end{equation*}
It is clear that is of type $(2,3;0,1;2,1,1,2)$ and has $|C^{\perp}|=2^{1}2^{8}2^{3}2^{2}2^{2}=2^{16}=65536$ codewords.
\end{example}

\section{The Structure of $\mathbb{Z}_{2}[u]\mathbb{Z}_{2}[u,v]$-additive cyclic code}
In this section, we introduce the definition of a additive cyclic code and some algebraic structure.

Let $R_{\alpha,\beta}[x]=\dfrac{R_{1}[x]}{<x^{\alpha}-1>}\times \dfrac{R_{2}[x]}{<x^{\beta}-1>}.$

A code $C$ is cyclic if and only if
its polynomial representation is an ideal. 

A additive code $C$ is called a \textbf{$\mathbb{Z}_{2}[u]\mathbb{Z}_{2}[u,v]$-additive cyclic code} if any cyclic shift of a codeword is 
also a codeword. 

i.e.,
 $(a_{0},a_{1},\cdots,a_{\alpha-1},b_{0},b_{1},\cdots,b_{\beta-1})\in C $ 
 
$\Rightarrow (a_{\alpha-1},a_{0},
 \cdots,a_{\alpha-2},b_{\beta-1},b_{0},\cdots,b_{\beta-2})\in C.$

 \begin{theorem}
 If $C$ is any $\mathbb{Z}_{2}[u]\mathbb{Z}_{2}[u,v]$-additive cyclic code, then $C^{\perp}$ is also cyclic.
\end{theorem}
\begin{proof}
 Let $C$ be any $\mathbb{Z}_{2}[u]\mathbb{Z}_{2}[u,v]$-additive cyclic code.
 
 Suppose
 $x=(a_{0},a_{1},\cdots,a_{\alpha-1},b_{0},b_{1},\cdots,b_{\beta-1})\in C^{\perp},$ 
 
 for any codeword 
 
 $y=(d_{0},d_{1},\cdots,d_{\alpha-1},e_{0},e_{1},\cdots,e_{\beta-1})\in C$
 
 we have
 
 $\langle x,y\rangle=uv\left(\sum\limits_{i=0}^{\alpha-1}a_{i}d_{i}\right)+\sum\limits_{j=0}^{\beta-1}b_{j}e_{j}=0.$
 
 Let $S$ is a cyclic shift, and $j=lcm(\alpha,\beta).$ Then we have
 
 $S(x)=(a_{\alpha-1},a_{0},
 \cdots,a_{\alpha-2},b_{\beta-1},b_{0},\cdots,b_{\beta-2})$ and $S^{j}(y)=y$ for any $y\in C.$
 
 Since $C$ be any 
 $\mathbb{Z}_{2}[u]\mathbb{Z}_{2}[u,v]$-additive cyclic code, So we have
 
 $S^{j-1}(y)=(d_{1},d_{2},\cdots,d_{\alpha-1},d_{0},e_{1},e_{2},\cdots,e_{\beta-1},e_{0})\in C.$
 
 Hence
 \begin{equation*}
  \begin{split}
   0&=\langle x, S^{j-1}(y)\rangle\\
   &=uv(a_{0}d_{1}+a_{1}d_{2}+\cdots+a_{\alpha-2}d_{\alpha-1}+a_{\alpha-1}d_{0})\\
   &+(b_{0}e_{1}+b_{1}e_{2}+\cdots b_{\beta-2}e_{\beta-1}+b_{\beta-1}e_{0})\\
   &=uv(a_{\alpha-1}d_{0}+a_{0}d_{1}+a_{1}d_{2}+\cdots+a_{\alpha-2}d_{\alpha-1})\\
   &+(b_{\beta-1}e_{0}+b_{0}e_{1}+b_{1}e_{2}+\cdots b_{\beta-2}e_{\beta-1})\\
   &=\langle S(x),y\rangle
  \end{split}
 \end{equation*}
Therefore, we have $S(x)\in C^{\perp},$ so $C^{\perp}$ is a cyclic code.
\end{proof}

Let $C$ be a $\mathbb{Z}_{2}[u]\mathbb{Z}_{2}[u,v]$-additive cyclic code, for any codeword 

$c=(a_{0},a_{1},\cdots,a_{\alpha-1},b_{0},b_{1},\cdots,b_{\beta-1})\in C$ can be representation with a polynomial, that is,

$c(x)=(a(x),b(x))\in R_{\alpha,\beta}[x]$

Similarly, we introduce a new scalar multiplication. Now, we have the following scalar multiplication:

for $c_{1}(x)=(a_{1}(x),b_{1}(x))\in R_{\alpha,\beta}[x]$ and $h(x)=p(x)+uq(x)+vr(x)+ uv s(x)\in R_{2}[x],$
define

$h(x)c_{1}(x)=((p(x)+uq(x))a_{1}(x),h(x)b_{1}(x)).$

Now, we define the homomorphism mapping:
$$\Psi: R_{\alpha,\beta}[x]\rightarrow R_{2}[x]$$

$$\Psi(c(x))=\Psi(a(x),b(x))=b(x).$$

It is clear that Image($\Psi$) is an ideal in the ring $\dfrac{R_{2}[x]}{<x^{\beta}-1>}$ and ker($\Psi$) is also an ideal
in $\dfrac{R_{1}[x]}{<x^{\alpha}-1>}.$
And note that

$Image(\Psi)=\left\{b(x)\in R_{2}[x] : (a(x), b(x))\in R_{\alpha,\beta}[x]\right\}$

$ker(\Psi)=\left\{(a(x),0)\in R_{\alpha,\beta}[x] : a(x)\in \dfrac{R_{1}[x]}{<x^{\alpha}-1>} \right\}.$

We have

$Image(\Psi)=<g(x)+up_{1}(x)+vp_{2}(x)+uvp_{3}(x),ua_{1}(x)+vq_{1}(x)+uvq_{2}(x), va_{2}(x)+uvr_{1}(x),uva_{3}(x)>$

where $g(x),p_{1}(x),p_{2}(x),p_{3}(x),a_{1}(x),q_{1}(x),q_{2}(x),
a_{2}(x),r_{1}(x),a_{3}(x)\in \dfrac{R_{2}[x]}{<x^{\beta}-1>} $ and 

$a_{3}(x)\,|\,a_{2}(x)|a_{1}(x)\,|\,g(x)\,|\,x^{\beta}-1 \,\, mod\,\, 2,$

$ a_{1}(x)\,|\,p_{1}(x)\left(\dfrac{x^{\beta}-1}{g(x)}\right),$

$a_{2}(x)\,|\,q_{1}(x)\left(\dfrac{x^{\beta}-1}{a_{1}(x)}\right),$ 

$ a_{2}(x)\,|\,p_{2}(x)\left(\dfrac{x^{\beta}-1}{g(x)}\right)\left(\dfrac{x^{\beta}-1}{a_{1}(x)}\right),$

$a_{3}(x)\,|\,r_{1}(x)\left(\dfrac{x^{\beta}-1}{a_{2}(x)}\right),$

$a_{3}(x)\,|\,q_{2}(x)\left(\dfrac{x^{\beta}-1}{q_{1}(x)}\right)\left(\dfrac{x^{\beta}-1}{a_{1}(x)}\right).$

$a_{3}(x)\,|\,p_{3}(x)\left(\dfrac{x^{\beta}-1}{g(x)}\right)\left(\dfrac{x^{\beta}-1}{a_{2}(x)}\right)\left(\dfrac{x^{\beta}-1}{a_{1}(x)}\right).$

Similarly, 
$$ker(\Psi)=<((f(x)+ul(x),ub(x)),0)>$$

where
$f(x),l(x),b(x)\in \dfrac{R_{1}[x]}{<x^{\alpha}-1>}$ and

$b(x)\,|\,f(x)\,|\,x^{\alpha}-1$, 
$b(x)\,|\,l(x)\left(\dfrac{x^{\alpha}-1}{f(x)}\right)$. 

According to the fundamendal theorem of homomorphism, we have

 $\dfrac{C}{ker(\Psi)}\cong <g(x)+up_{1}(x)+vp_{2}(x)+uvp_{3}(x), ua_{1}(x)+vq_{1}(x)+uvq_{2}(x),
va_{2}(x)+uvr_{1}(x),uva_{3}(x)>.$

Hence, we have 
$((h(x)+ut(x),uc(x)),(g(x)+up_{1}(x)+vp_{2}(x)+uvp_{3}(x),ua_{1}(x)+vq_{1}(x)+uvq_{2}(x), va_{2}(x)+uvr_{1}(x),uva_{3}(x)))\in C,$

where $\Psi((h(x)+ut(x),uc(x)),(g(x)+up_{1}(x)+vp_{2}(x)+uvp_{3}(x),ua_{1}(x)+vq_{1}(x)+uvq_{2}(x), va_{2}(x)+uvr_{1}(x),uva_{3}(x)))=
(g(x)+up_{1}(x)+vp_{2}(x)+uvp_{3}(x),ua_{1}(x)+vq_{1}(x)+uvq_{2}(x), va_{2}(x)+uvr_{1}(x),uva_{3}(x)).$

By these discussion, it is easy to see that any $\mathbb{Z}_{2}[u]\mathbb{Z}_{2}[u,v]$-additive cyclic code can be generated by two elements of the 
form $((h(x)+ut(x),uc(x)),(g(x)+up_{1}(x)+vp_{2}(x)+uvp_{3}(x),ua_{1}(x)+vq_{1}(x)+uvq_{2}(x), va_{2}(x)+uvr_{1}(x),uva_{3}(x)))$ and
$((f(x)+ul(x),ub(x)),0).$

\begin{corollary}
 Let $C$ be a $\mathbb{Z}_{2}[u]\mathbb{Z}_{2}[u,v]$-additive cyclic code. Then $C$ is an ideal in $R_{\alpha,\beta}[x]$ which can be generated by
 
  $C=(((f(x)+ul(x),ub(x)),0),((h(x)+ut(x),uc(x)),  (g(x)+up_{1}(x)+vp_{2}(x)+uvp_{3}(x),
  ua_{1}(x)+vq_{1}(x)+uvq_{2}(x),va_{2}(x)+uvr_{1}(x),uva_{3}(x)))),$
  
 where 
$a_{3}(x)\,|\,a_{2}(x)\,|\,a_{1}(x)\,|\,g(x)\,|\,x^{\beta}-1 \,\, mod\,\, 2,$

$ a_{1}(x)\,|\,p_{1}(x)\left(\dfrac{x^{\beta}-1}{g(x)}\right),$

$a_{2}(x)\,|\,q_{1}(x)\left(\dfrac{x^{\beta}-1}{a_{1}(x)}\right),$ 

$ a_{2}(x)\,|\,p_{2}(x)\left(\dfrac{x^{\beta}-1}{g(x)}\right)\left(\dfrac{x^{\beta}-1}{a_{1}(x)}\right),$

$a_{3}(x)\,|\,r_{1}(x)\left(\dfrac{x^{\beta}-1}{a_{2}(x)}\right),$

$a_{3}(x)\,|\,q_{2}(x)\left(\dfrac{x^{\beta}-1}{q_{1}(x)}\right)\left(\dfrac{x^{\beta}-1}{a_{1}(x)}\right).$

$a_{3}(x)\,|\,p_{3}(x)\left(\dfrac{x^{\beta}-1}{g(x)}\right)\left(\dfrac{x^{\beta}-1}{a_{2}(x)}\right)\left(\dfrac{x^{\beta}-1}{a_{1}(x)}\right),$

 $b(x)\,|\,f(x)|x^{\alpha}-1$, $b(x)\,|\,l(x)\left(\dfrac{x^{\alpha}-1}{f(x)}\right),$
 
 $c(x)\,|\,h(x)\,|\,x^{\alpha}-1$ and
 $c(x)\,|\,t(x)\left(\dfrac{x^{\alpha}-1}{h(x)}\right).$
\end{corollary}

\section{$(1+u,1+uv)$-additive constacyclic codes over $\mathbb{Z}_{2}[u]\mathbb{Z}_{2}[u,v]$}
 In this section, we introduce the definition of a additive constacyclic code and some algebraic structure.

 For fixed units $1+u \in R_{1}$ and $1+uv\in R_{2}$. The shift operator $\tau$ 
 on $R_{1}^{\alpha}\times R_{2}^{\beta}$ is defined to be
 
 \begin{equation*}
 \tau(a_{0},a_{1},\cdots,a_{\alpha-1},b_{0},b_{1},\cdots,b_{\beta-1})=\\((1+u)a_{\alpha-1},a_{0},\cdots,a_{\alpha-2}, (1+uv)b_{\beta-1},
 b_{0},\cdots,b_{\beta-2}).
 \end{equation*}

 A $\mathbb{Z}_{2}[u]\mathbb{Z}_{2}[u,v]$-additive code $C$ is said to be \textbf{$(1+u,1+uv)$-additive constacyclic code}
 if the code is inveriant under the
  shift $\tau,$ that is $\tau(C)=C.$

 Let $\sigma$ be the cyclic shift. For any positive integer $s$, let $\sigma_{s}$ be the quasi-shift given by
 
 $\sigma_{s}(a^{(1)}|a^{(2)}|\cdots|a^{(s)})=\sigma(a^{(1)})|\sigma(a^{(2)})|\cdots|\sigma(a^{(s)}),$
 
 where$a^{(1)},a^{(2)},\cdots,a^{(s)}\in \mathbb{Z}_{2}^{n}$ and $ | $ denotes the usual vector concatenation. 
 A binary quasi-cyclic code $D$ of index $s$ and lenght $n$ is a subset of $\mathbb{Z}_{2}^{n}$ such that $\sigma_{s}(D)=D.$
 \begin{lemma}\label{d}
  Let $\Psi$ be a Gray map defined as above and let $\tau$ be the $(1+u,1+uv)$-constacyclic shift on 
  $R_{1}^{\alpha}\times R_{2}^{\beta}.$ Then $\Psi \tau=\sigma_{2}\Psi.$
 \end{lemma}
 \begin{proof}
  Let $w=(x,y)=(x_{0},x_{1},\cdots,x_{\alpha-1},y_{0},y_{1},\cdots,y_{\beta-1})\in R_{1}^{\alpha}\times R_{2}^{\beta}.$
  
  Let $x_{i}=r_{i}+uq_{i}$ and $y_{j}=a_{j}+ub_{j}+vc_{j}+uvd_{j}$ where $x=(x_{0},x_{1},\cdots,x_{\alpha-1})$ and
  $y=(y_{0},y_{1},\cdots,y_{\beta-1})$ for 
  $0\leq i\leq \alpha-1,\,\, 0\leq j\leq \beta-1.$ From definitions, we have

  \begin{equation*}
 \begin{split}
  \Phi(w)&=(q_{0},q_{1},\cdots,q_{\alpha-1}, q_{0}+r_{0},q_{1}+r_{1},\cdots,q_{\alpha-1}+r_{\alpha-1},\\
  &d_{0},d_{1},\cdots,d_{\beta-1}, a_{0}+d_{0},a_{1}+d_{1},\cdots,a_{\beta-1}+d_{\beta-1},\\
 &b_{0}+d_{0},b_{1}+d_{1},\cdots,b_{\beta-1}+d_{\beta-1},\\
    & a_{0}+b_{0}+d_{0},a_{1}+b_{1}+d_{1},\cdots,a_{\beta-1}+b_{\beta-1}+d_{\beta-1},\\
  &c_{0}+d_{0},c_{1}+d_{1},\cdots,c_{\beta-1}+d_{\beta-1},\\
  &a_{0}+c_{0}+d_{0},a_{1}+c_{1}+d_{1},\cdots,a_{\beta-1}+c_{\beta-1}+d_{\beta-1},\\
  &b_{0}+c_{0}+d_{0},b_{1}+c_{1}+d_{1},\cdots,b_{\beta-1}+c_{\beta-1}+d_{\beta-1},\\ 
  &a_{0}+b_{0}+c_{0}+d_{0},\cdots,a_{\beta-1}+b_{\beta-1}+c_{\beta-1}+d_{\beta-1}),  
 \end{split}
 \end{equation*}
 hence,
 \begin{equation*}
  \begin{split}
 \sigma_{2}\Psi(w)&=(q_{\alpha-1}+r_{\alpha-1},q_{0},q_{1},\cdots,q_{\alpha-1},\\
 &q_{0}+r_{0},q_{1}+r_{1},\cdots,q_{\alpha-2}+r_{\alpha-2},\\
 & a_{\beta-1}+d_{\beta-1},d_{0},d_{1},\cdots,d_{\beta-1},\\
 &a_{0}+d_{0},a_{1}+d_{1},\cdots,a_{\beta-2}+d_{\beta-2},\\
 &a_{\beta-1}+b_{\beta-1}+d_{\beta-1},b_{0}+d_{0},b_{1}+d_{1},\cdots,\\&b_{\beta-1}+d_{\beta-1},
 a_{0}+b_{0}+d_{0},a_{1}+b_{1}+d_{1},\cdots,\\&a_{\beta-2}+b_{\beta-2}+d_{\beta-2},
 a_{\beta-1}+c_{\beta-1}+d_{\beta-1},c_{0}+d_{0},\\&c_{1}+d_{1},\cdots,c_{\beta-1}+d_{\beta-1},\\
 & a_{0}+c_{0}+d_{0},a_{1}+c_{1}+d_{1},\cdots,a_{\beta-2}+c_{\beta-2}+d_{\beta-2}\\
 &a_{\beta-1}+b_{\beta-1}+c_{\beta-1}+d_{\beta-1},b_{0}+c_{0}+d_{0},\\
 &b_{1}+c_{1}+d_{1},\cdots,b_{\beta-1}+c_{\beta-1}+d_{\beta-1},\\ 
  &a_{0}+b_{0}+c_{0}+d_{0},a_{1}+b_{1}+c_{1}+d_{1},\cdots,\\&a_{\beta-2}+b_{\beta-2}+c_{\beta-2}+d_{\beta-2}),  
  \end{split}
 \end{equation*}
On the other hand,
\begin{equation*}
 \begin{split}
  \tau(w)&=((1+u)x_{\alpha-1},x_{0},\cdots,x_{\alpha-2},\\&(1+uv)y_{\beta-1},y_{0},\cdots,y_{\beta-2})\\
  &=( r_{\alpha-1}+u(r_{\alpha-1}+q_{\alpha-1}),r_{0}+uq_{0},\cdots,\\&r_{\alpha-2}+uq_{\alpha-2},\\
 & a_{\beta-1}+ub_{\beta-1}+vc_{\beta-1}+uv(a_{\beta-1}+d_{\beta-1}),\\
  &a_{0}+ub_{0}+vc_{0}+uvd_{0},\cdots,\\&a_{\beta-2}+ub_{\beta-2}+vc_{\beta-2}+uvd_{\beta-2}).
 \end{split}
 \end{equation*}
 Hence,
 \begin{equation*}
  \begin{split}
   \Psi\tau(w)&=(q_{\alpha-1}+r_{\alpha-1},q_{0},q_{1},\cdots,q_{\alpha-2},q_{\alpha-1},\\&q_{0}+r_{0},q_{1}+r_{1},\cdots,q_{\alpha-2}+r_{\alpha-2},\\
   & a_{\beta-1}+d_{\beta-1},d_{0},d_{1},\cdots,d_{\beta-1},a_{0}+d_{0},\\&a_{1}+d_{1},\cdots,a_{\beta-2}+d_{\beta-2},\\
 &a_{\beta-1}+b_{\beta-1}+d_{\beta-1},b_{0}+d_{0},b_{1}+d_{1},\cdots,b_{\beta-1}+d_{\beta-1},\\
 &a_{0}+b_{0}+d_{0},a_{1}+b_{1}+d_{1},\cdots,a_{\beta-2}+b_{\beta-2}+d_{\beta-2},\\
 &a_{\beta-1}+c_{\beta-1}+d_{\beta-1},c_{0}+d_{0},c_{1}+d_{1},\cdots,c_{\beta-1}+d_{\beta-1},\\
 & a_{0}+c_{0}+d_{0},a_{1}+c_{1}+d_{1},\cdots,a_{\beta-2}+c_{\beta-2}+d_{\beta-2}\\
 &a_{\beta-1}+b_{\beta-1}+c_{\beta-1}+d_{\beta-1},b_{0}+c_{0}+d_{0},\\&b_{1}+c_{1}+d_{1},\cdots,b_{\beta-1}+c_{\beta-1}+d_{\beta-1},\\ 
  &a_{0}+b_{0}+c_{0}+d_{0},a_{1}+b_{1}+c_{1}+d_{1},\cdots,\\&a_{\beta-2}+b_{\beta-2}+c_{\beta-2}+d_{\beta-2}).   
 \end{split}
 \end{equation*}
 \end{proof}

\begin{theorem}
 A $\mathbb{Z}_{2}[u]\mathbb{Z}_{2}[u,v]$-additive code $C$ is a $(1+u,1+uv)$-additive constacyclic code if and only if 
 $\Psi(C)$ is a binary quasi-cyclic code of index $2.$
\end{theorem}
\begin{proof}
 If $C$ is $(1+u,1+uv)$-additive constacyclic, then using Lemma(\ref{d}) we have
 $$\sigma_{2}(\Psi(C))=\Psi(\tau(C))=\Psi(C).$$
 Hence, $\Psi(C)$ is a binary quasi-cyclic code of index $2.$ Conversely, if $\Psi(C)$ is a binary quasi-cyclic code of index $2,$
 then using Lemma (\ref{d}) again we get
$$\Psi(\tau(C))=\sigma_{2}(\Psi(C))=\Psi(C).$$
\end{proof}
\begin{corollary}
 The image of a $(1+u,1+uv)$-additive constacyclic code over $R_{1}^{\alpha}\times R_{2}^{\beta}$ under the Gray 
 map $\Psi$ is a distance invariant binary quasi-cyclic code of index $2.$
\end{corollary}

\section{Conclusion}
In this paper, we studied $\mathbb{Z}_{2}[u]\mathbb{Z}_{2}[u,v]$-additive codes some property, including generator and parity check matrices for the 
codes. We fund the Gray map $\Phi$ is a distance preserving map and weight preserving map as well. At the end of this paper, we introduce
the structure of  $\mathbb{Z}_{2}[u]\mathbb{Z}_{2}[u,v]$-additive cyclic codes and constacyclic codes.

\end{document}